\documentclass[journal]{IEEEtran}
\ifCLASSINFOpdf
\else
\fi
\usepackage[cmex10]{amsmath}
\usepackage{graphicx}
\usepackage{cite}
\usepackage{amssymb}

\newtheorem{theorem}{Theorem}
\newtheorem{lemma}{Lemma}
\newtheorem{corollary}{Corollary}
\newtheorem{remark}{Remark}
\begin{document}
\title{Optimal Tracking Performance Limitation of Networked Control Systems with Limited \\Bandwidth and Additive Colored White \\ Gaussian Noise}
\author{Zhi-Hong~Guan,
        Chao-Yang~Chen,
        Gang~Feng,
        ~and~Tao Li
\thanks{Zhi-Hong Guan and Chao-Yang Chen are with the Department of Control Science and Engineering, Huazhong University of Science and Technology, Wuhan, 430074, P. R. China.}

\thanks{Gang Feng is with the Department of Mechanical and Biomedical Engineering, City University of Hong Kong, Kowloon, Hong Kong, SAR, P. R. China.}

\thanks{Tao Li is with the College of Electronics and Information, Yangtze University, Jingzhou, 434023, P. R. China.}

\thanks {Corresponding Author: zhguan@mail.hust.edu.cn (Z.H. Guan).}

\thanks{This work was supported in part by the National Natural Science Foundation of China under Grants 60834002, 60973012, 61073065, 61100076 and the Doctoral Foundation of Ministry of Education of China under Grant 20090142110039 and by a grant from the Research Grants Council of the Hong Kong Special Administrative Region, China (Project No.: CityU 113311).}

}

\markboth{IEEE TRANSACTIONS ON CIRCUITS AND SYSTEMS¡ªI: REGULAR PAPERS}%
{Shell \MakeLowercase{\textit{et al.}}: Bare Demo of IEEEtran.cls for Journals}

\maketitle
\begin{abstract}
This paper studies optimal tracking performance issues for multi-input-multi-output linear time-invariant systems under networked control with limited bandwidth and additive colored white Gaussian noise channel. The tracking performance is measured by control input energy and the energy of the error signal between the output of the system and the reference signal with respect to a Brownian motion random process. This paper focuses on two kinds of network parameters, the basic network parameter-bandwidth and the additive colored white Gaussian noise, and studies the tracking performance limitation problem. The best attainable tracking performance is obtained, and the impact of limited bandwidth and additive colored white Gaussian noise of the communication channel on the attainable tracking performance is revealed. It is shown that the optimal tracking performance depends on nonminimum phase zeros, gain at all frequencies and their directions unitary vector of the given plant, as well as the limited bandwidth and additive colored white Gaussian noise of the communication channel. The simulation results are finally given to illustrate the theoretical results.
\end{abstract}

\IEEEpeerreviewmaketitle

\section{Introduction}
\IEEEPARstart{M}{ore}
and more researchers are interested in networked control systems in the past decade, please see, for example, \cite{braslavsky2007feedback,li2009optimal,xisheng2010performance,rojas2008fundamental,xiao2010feedback,menon2010static,guan2011optimal} and references therein. Most works focus on analysis and synthesis of networked control systems with quantization effects (e.g. \cite{Zhang2011Quantized,Azuma2012Dynamic,qi2009optimal,you2009optimality,Xiao2010stabilization}), time delays\cite{Luan2011Stabilization,Wei2009Filtering,Liu2010Predictive}, bandwidth constraint\cite{rojas2006output,rojas2008fundamental,Trivellato2010State}, data rate constraint \cite{rojas2006output,braslavsky2007feedback}, and/or data packet dropout\cite{Trivellato2010State,wu2007design,Wang2011a,you2011mean}.
In spite of the significant progress in those studies, the more inspiring and challenging issues of control performance limitation under such network environment remain largely open.
\par
Performance limitations resulting from nonminimum phase (NMP) zeros and unstable poles of given systems have been known for a long time. The issue has been attracting a growing amount of interest in the control community, see   \cite{toker2002tracking,chen2003best,bakhtiar2008regulation,wang2011optimal} for example. The tracking performance achievable via feedback was studied in \cite{morari1989robust} with respect to single-input-single-output (SISO) stable systems.  The result was extended to multi-input-multi-output (MIMO) unstable systems in \cite{chen2000limitations}, and it was found that the minimal tracking error depends not only on the location of the system nonminimum phase zeros, but also on how the input signal may interact with those zeros, i.e., the angles between the input and zero directions. Optimal tracking and regulation control problems were studied in \cite{chen2003best}, where objective functions of tracking error and regulated response, defined by integral square measures, are minimized jointly with the control effort, and the latter is measured by the system input energy. In \cite{wang2011optimal}, the optimal tracking control problem was studied with both the forward and feedback channel disturbances. The authors of \cite{bakhtiar2008regulation} investigated the regulation performance limitations of unstable non-minimum phase single-input-multi-output (SIMO) continuous-time and discrete-time systems, respectively. However, all these mentioned works have not taken into account the effects of networks, which would make the study of the optimal performance limitation much more challenging.
\par
Networked control systems are ubiquitous in industry. More and more control systems are operating over a network. In recent years, the research on the performance limitation of networked control systems attracts some attention. For example, the authors in \cite{qi2009optimal} studied the tracking performance of discrete-time SISO networked feedback systems, by modeling the quantization error as a white noise. The tracking performance of continuous-time MIMO  systems with the additive white Gaussian noise (AWGN) was studied through one- and two-parameter control schemes in \cite{ding2010tracking,Zhan2012Optimal}. The result was further generalized to other noisy channels with bandwidth limitation in [8], where the optimal tracking performance is measured by the achievable minimal tracking error. However, it was showed in [24] that, in the optimal tracking problem, in order to attain the minimal tracking error, the control input of systems is often required to have an infinite energy. This requirement cannot be met in general in practice. Thus the control input energy of systems should be considered in the performance index to address this issue. In this paper, we consider the optimal tracking problem in terms of both the  tracking error energy and the control input energy, and meanwhile we consider communication link over bandwidth-limited additive colored Gaussian noise (ACGN) channels, which are more realistic models of communication link than those in \cite{braslavsky2007feedback,wang2011optimal}.
\par
In this paper, we study optimal tracking performance issues pertaining to MIMO feedback control systems. The objective is to minimize the tracking error between the output and the reference signals of a feedback system under the constraint of control input energy. The optimal tracking performance is attained by stabilizing compensators under a two-parameter structure. The tracking error is defined in an square error sense, and the reference signals are considered as a Brownian motion, which can be roughly considered as the integral of a standard white noise \cite{qiu2002fundamental,li2009optimal,ding2010tracking}. The tracking performance index is given by the weighted sum between the power of the tracking error energy and the system input energy.
\par
The rest of the the paper is organized as follows. The problem formulation and preliminaries are given in section II. In section III, the main results of this paper are presented. Results of extensive simulation studies and discussions are shown to validate the theoretical results in section IV. Concluding remarks are made in Section V.

\section{Preliminaries}
We begin by summarizing briefly the notations used throughout this paper. For any complex number $s$, we denote its complex conjugate by $s^H$. The expectation operator is denoted by $E\{\cdot\}$, respectively. For any vector $u$, we denote its conjugate transpose by $u^H$, and its Euclidean norm by $\left\| u \right\|$. For a matrix $A$, we denote its conjugate transpose by ${A^H}$. All the vectors and matrices involved in the sequel are assumed to have compatible dimensions, and for simplicity their dimensions will be omitted. Let the open right-half plane be denoted by $\mathbb{C}_+:=\{s:Re(s)>0\}$,
the open left-half plane by $\mathbb{C}_-:=\{s:Re(s)<0\} $, and the imaginary axis by $\mathbb{C}_0$. Define
${L_2}: = \{ f:f(s){\rm{~measurable~in~}}{\mathbb{C}_0}, {\| f \|_2^2}:=\frac{1}{{2\pi }}{\int_{-\infty}^\infty{\|{f(j\omega)} \|}_F^{2}}\rm{d}\omega<\infty\}.$ Then, ${L_2}$ is a Hilbert space with\\ an inner product $\langle{f,g}\rangle:=\frac{1}{{2\pi}}\int_{-\infty}^{\infty} \mathrm{tr}\{{f^H}(j\omega)g(j\omega)\}\rm{d}\omega.$ Next, define $H_2$ as a subspace of functions in $L_2$ with functions $f(s)$ analytic in $\mathbb{C}_-$, ${H_2}: =\{f:f(s){\rm{~analytic~in~C_,}}\\ {\|f\|_2^2}:=\sup_{\sigma>0}\frac{1}{{2\pi }}{\int_{-\infty}^\infty\|{f(\sigma+j\omega)}\|_F^{2}}\rm{d}\omega<\infty\}.$
and the orthogonal complement of $H_2$ in $L_2$ as $H_2^\bot$:
${H_2^\bot}:=\{f:f(s){\rm{analytic}}{\rm{in~C_,\,}}{\|f\|_2^2}:=\sup_{\sigma<0}\frac{1}{{2\pi }}{\int_{-\infty}^\infty\|{f(\sigma+j\omega)}\|_F^{2}}\rm{d}\omega\\<\infty\}.$
Thus, for any $f \in H_2^ \bot $ and $g \in {H_{2,}}\left\langle {f,g} \right\rangle = 0$.
We use the same notation ${\left\|\cdot\right\|_2}$ to denote the corresponding norm.
Finally, we denote by $\mathbb{R}{\mathcal{H}_\infty}$  the class of all stable, proper rational transfer function matrices.
We introduce a factorization formula for non-minimum phase systems.
For the right-invertible rational transfer function matrix $P$,
let its right and left coprime factorizations be given by
\begin{equation}\label{eqq1}
P=NM^{-1}=\tilde{M}^{-1}\tilde{N},
\end{equation}
where $N,M,\tilde{M},\tilde{N}\in \mathbb{R}H_{\infty}$.
A complex number $s\in \mathbb{C}$ is said to be a zero of $P(s)$, if ${\eta^H}P(s)=0$ for some unitary vector $\eta$,
where $\eta$ is called an output direction vector associated with $s$, and $\left\| \eta  \right\| = 1$.
For such a zero, it is always true that $\eta^HN(s)=0$, for some unitary vector $\eta$.
On the other hand, a complex number is said to be a pole of P(s) if $P(p) = \infty$.
If $p$ is an unstable pole of $P(s)$, i.e., $p \in {\mathbb{C}_+}$,
then equivalent statement is that $\tilde M(p)\omega = 0$ for some unitary vector $\omega$, $\left\| \omega \right\| = 1$.
In order to facilitate the subsequent proof, we introduce two specific factorization for $N(s): N(s)=L(s){N_m}(s)=\hat{L}(s){\hat{N}_m}(s).$
And, allpass factor $L(s)$ and $\hat{L}(s)$ have the form
$L(s):=\prod_{i=1}^{n_z}{L_i}(s),\hat{L}(s):=\prod_{i=1}^{n_z}{\hat{L}_i}(s),$
and
\begin{align}\label{eq2}
   {L_i}(s):=&{[\eta_i~~U_i]
\left[
  \begin{array}{cc}
    \frac{\bar{z}_i}{z_i}\frac{z_i-s}{\bar{z}_i+s} & 0 \\
    0 & I
  \end{array}
\right]
  \left[
    \begin{array}{c}
      \eta_i^H \\
      U_i^H
    \end{array}
  \right]},\\
  \label{eq3}{\hat{L}_i}(s):=&{[\hat{\eta}_i~~\hat{U}_i]
\left[
  \begin{array}{cc}
    \frac{s-z_i}{s+\bar{z}_i} & 0 \\
    0 & I
  \end{array}
\right]
  \left[
    \begin{array}{c}
      \hat{\eta}_i^H \\
      \hat{U}_i^H
    \end{array}
  \right]},
\end{align}
where ${\eta_i}$ are unitary vectors obtained by factorizing the zeros one at a time, and ${U_i}$ are matrices which together with ${\eta _i}$ form a unitary matrix. Similarly, ${\hat{\eta}_i}$ and ${\hat{U}_i}$ have same definition and nature.\par
Likewise, $\tilde{M}$ has the allpass factorization
$\tilde{M}=\tilde{M}_m(s)$ $\times\tilde{B}(s),$
where $\tilde{B}(s)$ is an allpass factor and $\tilde{M}_m(s)$ is the minimum phase part of $\tilde{M}(s)$. One particular allpass factor is given by
$\tilde{B}(s):=\prod_{i=1}^{n_p}{\tilde{B}_i}(s),$
and
\begin{align}\label{BBi}
  {\tilde{B}_i}(s):=&{[\tilde{\omega}_i~~\tilde{W}_i]
\left[
  \begin{array}{cc}
    \frac{s-p_i}{s+\bar{p}_i} & 0 \\
    0 & I
  \end{array}
\right]
  \left[
    \begin{array}{c}
      \tilde{\omega}_i^H \\
      \tilde{W}_i^H
    \end{array}
  \right]}.
\end{align}
Consider the class of functions in
$
\mathbb{F}:=\{f:f(s)~{\rm{analytic~in}}\\\mathbb{C}_+,  \lim_{R\rightarrow\infty}\max_{\theta\in[-\pi/2,\pi/2]}\|f(Re^{j\theta})/{R}\|=0\}.
$. Lemma \ref{le1} and \ref{le2} can be found in \cite{chen2003best}.
\begin{lemma} \label{le1}
Let $f(s)\in\mathbb{F}$ and denote
$f(j\omega)=h_1(\omega)+jh_2(\omega).$
Suppose that $f(s)$ is conjugate symmetric, i.e., $f(s)=\overline{f(\bar{s})}$. Then
 $f^\prime(0)=({1}/{\pi})\int_\infty^\infty\big({h_1(\omega)-h_1(0)}\big)/{\omega^2}\rm{d}\omega\,.$
\end{lemma}
\begin{lemma}\label{le2}
Consider a conjugate symmetric function $f(s)$. Suppose that $f(s)$ is analytic and has no zero in $\mathbb{C}_+$, and that $\log{f(s)}\in\mathbb{F}$. Then
provided that ${f^\prime(0)}/{f(0)}=({1}/{\pi})\int_{-\infty}^{+\infty}
({1}/{\omega^2})\log|{f(j\omega)}/{f(0)}|\mathrm{d}\omega\,,f(0)\neq0.$
 \end{lemma}
\begin{lemma}\label{le3}
Let $L$ and $L_i$ be defined by (\ref{eq3}). Then, for any $X\in{\mathbb{R}{\mathcal{H}_\infty}}$, the equality
$ XL^{-1}=S+\sum_{i=1}^{N_z}X(z_i)L_1^{-1}(z_i)\\
\cdots L_{i-1}^{-1}(z_i)L_i^{-1}L_{i+1}^{-1}(z_i)\cdots L_{N_z}^{-1}(z_i)$
 holds for some $S\in \mathbb{R}\mathcal{H}_\infty$.
\end{lemma}
\begin{proof}
We assume that $A$ is an allpass factor.
From lemma 4.1 in \cite{wang2009limitations}, for some $Y\in{\mathbb{R}\mathcal{H}_\infty}$, we have $A^{-1}Y=S_1+\sum_{i=1}^{N_z}A_{N_z}^{-1}(z_i)\cdots A_{i+1}^{-1}(z_i)
A_i^{-1}A_{i-1}^{-1}(z_i)\cdots A_1^{-1}(z_i)Y(z_i).$
Then, we have
$Y^HA^{-H}=S_1^H+\sum_{i=1}^{N_z}Y^H(z_i)A_1^{-H}(z_i)\cdots\\ A_{i-1}^{-H}(z_i) A_i^{-H}A_{i+1}^{-H}(z_i)\cdots A_{N_z}^{-H}(z_i).$
Let $L=A^H, S=S_1^H$, $X=Y^H,$ then
$XL^{-1}=S+\sum_{i=1}^{N_z}X(z_i)L_1^{-1}(z_i)L_{i-1}^{-1}(z_i)\cdots\\ L_i^{-1}
L_{i+1}^{-1}(z_i)\cdots L_{N_z}^{-1}(z_i).$
Therefore, the proof is completed.
\end{proof}
\section{Tracking performance limitations}
Consider the control feedback loop shown in Fig.1, where the plant model $P$ is a rational transfer function matrix. The channel model is the bandwidth-limited ACGN channel, where  $n=[n_1,n_2,\cdots,n_l]$ with $n_i(1\leq{i}\leq{l})$ being a zero-mean stationary white Gaussian noise process and spectral density $\gamma_i^2$ (when $l=1$, note $\gamma=\gamma_1$). The reference signal $r$ is a vector of the step signal generated by passing a standard white noise $\omega$ through an integrator, which can be roughly considered as a Brownian motion process\cite{qiu2002fundamental} and emulate the step signal in the deterministic setting\cite{li2009optimal,ding2010tracking}. Therefore, the formulation resembles the tracking of a deterministic step signal. For the channel $i$, we denote the spectral density of $w_i$ by $\sigma_i^2$ (when $l=1$, note $\sigma=\sigma_1$). It is assumed that the system reference inputs in different channels are independent, and that the reference input and the noise are uncorrelated.
\begin{figure}[h]
\centering
  \includegraphics[width=8.5cm]{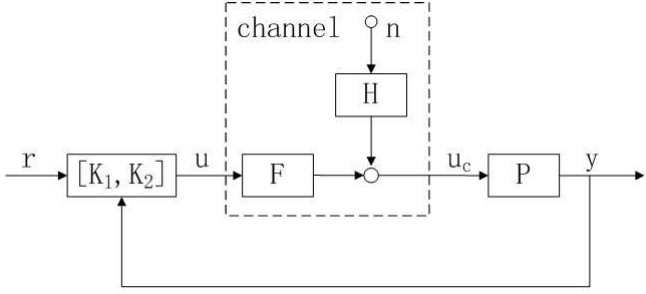}
  \caption{Feedback control over bandwidth limited ACGN channels}\label{Fig.1}
\end{figure}
 $[K_1,~K_2]$ denotes the two-parameter compensators. The communication channel is characterized by three parameters: the AWGN n, the channel transfer functions F and H. The channel transfer function $F(s)\in{\mathbb{R}\mathcal{H}_\infty}$ modeling the bandwidth limitation is assumed to be stable and NMP \cite{rojas2008fundamental}. Then
$F(s)=diag[ f_1(s),f_2(s), \cdots , f_l(s)],$
where $f_1(s)=f_2(s)=\cdots=f_l(s)$.  $F(s)$ has $n_f$ distinct NMP zeros. The channel transfer function $H(s)\in{\mathbb{R}\mathcal{H}_\infty}$, colors the additive white Gaussian noise.
 The performance index of the system is defined as
\begin{equation}\label{eq5}
    J=E[(1-\epsilon)(r(t)-y(t))^T(r(t)-y(t))+\epsilon u_c^T(t)u_c(t)],
\end{equation}
where the parameter $\epsilon~(0\leq\epsilon\leq1)$ is pre-set and can be used to weigh the relative importance
of tracking objective and the plant control input energy constraint.
For the transfer function matrices $P$ and $PF$, let their right and left coprime factorizations respectively be given by
 $PF = N{M^{ - 1}} = {\tilde M^{ - 1}}\tilde N,~~P=N_0M^{-1},$
where $N,N_0,M,\tilde N,\tilde M \in \mathbb{R}{\mathcal{H}_\infty},$
and satisfy the double Bezout identity
\begin{equation}\label{eq6}
    \left[ {\begin{array}{cc}
{\tilde X }&{ - \tilde Y}\\
{ - \tilde N}&{\tilde M}
\end{array}} \right]\left[ {\begin{array}{*{20}{c}}
M&Y\\
N&X
\end{array}} \right] = I.
\end{equation}
Then the set of all stabilizing two parameter compensators is characterized by
\begin{multline*}
{\mathcal{K}_s}: = \{ K:K = [K_1~~K_2]= (\tilde{X} - R\tilde{N})^{-1}\\
\times[Q~~\tilde{Y}-R\tilde{M}], Q\in{\mathbb{R}\mathcal{H}_\infty},R\in{\mathbb{R}\mathcal{H}_\infty}\}.
\end{multline*}
According to (\ref{eq5}), we may rewrite the performance index $J$ as
\begin{multline}\label{eq7}
    J: = E\bigg\{ (1-\epsilon)({{\left\| {r(t)-y_r(t)} \right\|}^2}\\
    +{{\left\|{y_n(t)}\right\|}^2})+\epsilon{{\left\| {{u}_c(t)} \right\|}^2} \bigg\},
\end{multline}
where $y_{r}(t)$ and $y_{n}(t)$ are the outputs in response to $r$ and $n$,  respectively. And the tracking error $\hat{e}$ is given by
\begin{equation*}
\hat{e}(t):=r(t)-y_r(t).
\end{equation*}
The optimal performance attainable by all possible stabilizing controllers is
\begin{equation*}
    {J^*}: = \mathop {\inf }\limits_{k \in {K_s}} {\rm{ }}J.
\end{equation*}
\begin{theorem}\label{th1}
Let $\omega$ and $n$ be uncorrelated white Gaussian signals. Suppose that $P(s)=P_o(s)/s^{n},$ for some integer $n\geq{1}$, such that $P_o(s)$ is proper and has no zero at $s=0$. $P$ is supposed to be unstable, NMP and invertible (including right invertible and left invertible).
Denote the NMP zeros of $P(s)$ and $F(s)$ by $z_i,(i=1,\cdots,n_z+n_f)$ and assume also that these zeros are distinct. Define
$f(s):=\mathrm{tr}\{(1-\epsilon)U^TN_m(s)\\\Theta_o^{-1}(s)\Theta_o^{-T}(0)N_m^T(0)U\}$
and factorize $f(s)$ as $f(s):=\big(\prod_{i=1}^{N_s}{\bar{s}_i(s_i-s)}{(s_i(\bar{s}_i+s))}\big)f_m(s),$
where $s_i\in\mathbb{C}_+$ are the nonminimum phase zeros of $f(s)$ and $f_m(s)$ is minimum phase. It is noted that, $f(s),
f_m(s)\in\mathbb{R}\mathcal{H}\infty$, $ f(0)=f_m(0)=\sum_{i=1}^l\sigma_i^2.$
Then, with the two-parameter controller given in Fig.\ref{Fig.1}
\begin{align*}
J^*=&2(1-\epsilon)\Bigg[\sum_{i=1}^{n_z+n_f}\frac{{\mathrm{Re}}(z_i)}{|z_i|^2}
\sum_{j=1}^l\sigma_j^2\cos^2\angle(\eta_i,e_j)\\
&+(\sum_{i=1}^l{\sigma_i^2})\left(\sum_{i=1}^{N_s}\frac{\mathrm{Re}s_{i}}{|s_{i}|^2}
-\frac{1}{\pi}\int_{0}^{+\infty}\frac{\log{|f(j\omega)|}}{\omega^2}\mathrm{d}\omega\right)\Bigg]\\
&+\sum_{i,j=1}^{n_z+n_f}\frac{4Re(z_i)Re(z_j)}{\bar{z}_i+z_j}{\omega}_j^HD^r_i(z_j){D}_i^{rH}(z_i){\omega}_i\\
& \times{\omega}_i^H{D_i^l}^H(z_i)V^HO^H(z_i)O(z_j)VD_j^l(z_j){\omega}_j.
\end{align*}
\end{theorem}
\begin{proof}
From (\ref{eq7}), we have
\begin{multline}\label{eqq4}
J: = (1-\epsilon){\rm{tr}}\{{R_{\hat{e}_r}(0)}+{R_{y_n}(0)}\}\\
+\epsilon{\rm{tr}}\{{R_{u_{cr}}(0)}+{R_{u_{cn}}(0)}\},
\end{multline}
where $R_{\hat{e}_r}(t), R_{y_n}(t), R_{u_{cr}}(t)$ and $R_{u_{cn}}(t)$, are the autocorrelation functions of the random processes $\hat{e}_r(t), y_n(t), u_{cr}(t)$ and $u_{cn}(t)$, respectively. Denote the spectral densities of $r$ and $n$ as $S_r(j\omega)$ and $S_n(j\omega)$ respectively.
\par
Then we have
\begin{align*}
 J=&(1-\epsilon)\frac{1}{2\pi}\bigg[\int_{-\infty}^{+\infty}{\rm{tr}(T_{\hat{e}_r}S_r(j\omega)T_{\hat{e}_r}^T)}{\rm{d}}\omega\\
 &\qquad\qquad\qquad\qquad\quad+\int_{-\infty}^{+\infty}{\rm{tr}(T_{y_n}S_n(j\omega)T_{y_n}^T)}{\rm{d}}\omega\bigg]\\
&+\epsilon\frac{1}{2\pi}\bigg[\int_{-\infty}^{+\infty}{\rm{tr}(T_{u_{cr}}S_r(j\omega)T_{u_{cr}}^T)}{\rm{d}}\omega\\
 &\qquad\qquad\qquad\qquad\quad+\int_{-\infty}^{+\infty}{\rm{tr}(T_{u_{cn}}S_n(j\omega)T_{u_{cn}}^T)}{\rm{d}}\omega\bigg]\\
=&(1-\epsilon)\Bigg(\bigg\|T_{\hat{e}_r}U\frac{1}{s}\bigg\|_2^2+\bigg\|T_{y_n}V\bigg\|_2^2\Bigg)\\
&\qquad\qquad\qquad\qquad\quad+\epsilon\Bigg(\bigg\|T_{u_{cr}}U\frac{1}{s}\bigg\|_2^2+\bigg\|T_{u_{cn}}V\bigg\|_2^2\Bigg)
\end{align*}
\begin{align*}
=&(1-\epsilon)\Bigg(\bigg\|\left[I-(I-PFK_2)^{-1}PFK_1\right]U\frac{1}{s}\bigg\|_2^2\\&\qquad\qquad\qquad\qquad\qquad\quad+\bigg\|(I-PFK_2)^{-1}PHV\bigg\|_2^2\Bigg)\\
&+\epsilon\Bigg(\bigg\|(I-FK_2P)^{-1}FK_1U\frac{1}{s}\bigg\|_2^2\\&\quad\qquad\qquad\qquad\qquad\qquad+\bigg\|(I-FK_2P)^{-1}HV\bigg\|_2^2\Bigg)\\
=&\left\|
\begin{bmatrix}
    \sqrt{1-\epsilon}(I-NQ) \\
    \sqrt{\epsilon}FMQ \\
  \end{bmatrix}U\frac{1}{s}
\right\|_2^2\\
&\qquad\qquad\qquad\qquad+\left\|
\begin{bmatrix}
    \sqrt{1-\epsilon}PM({\tilde{X}}-R{\tilde{N}})\\
    \sqrt{\epsilon}M({\tilde X}-R{\tilde N}) \\
  \end{bmatrix}HV
\right\|_2^2\\
=&J_U+J_V,
\end{align*}
where
$$U=diag[\sigma_1,\sigma_2,\cdots,\sigma_l],
    ~V=diag[\gamma_1,\gamma_2,\cdots,\gamma_l].
$$
Evidently, we have
$$J^*=\inf_{K\in \mathcal{K}}J=\inf_{Q\in \mathbb{R}\mathcal{H}_\infty}J_U+\inf_{R\in \mathbb{R}\mathcal{H}_\infty}J_V=J_U^*+J_V^*.$$
Firstly, for $J_U$, using the allpass factorization (\ref{eq2}), we have
\begin{align*}
J_U^*=&\inf_{Q\in \mathbb{R}\mathcal{H}_\infty}\left\|
\begin{bmatrix}
    \sqrt{1-\epsilon}(I-NQ) \\
    \sqrt{\epsilon}FMQ \\
  \end{bmatrix}U\frac{1}{s}
\right\|_2^2\\
=&\inf_{Q\in \mathbb{R}\mathcal{H}_\infty}\left\|
\begin{bmatrix}
\sqrt{1-\epsilon}(L^{-1}-I)\\0
\end{bmatrix}U\frac{1}{s}
\right\|_2^2\\
&\qquad\qquad+
\left\|
\begin{bmatrix}
\sqrt{1-\epsilon}I\\0
\end{bmatrix}U\frac{1}{s}
+\begin{bmatrix}
    -\sqrt{1-\epsilon}N_m \\
    \sqrt{\epsilon}C_{FM} \\
  \end{bmatrix}QU\frac{1}{s}
\right\|_2^2\\
=&2(1-\epsilon)\sum_{i=1}^{n_z+n_f}\frac{{\mathrm{Re}}(z_i)}{|z_i|^2}\|\eta_i^H U\|_F^2+\inf_{Q\in \mathbb{R}\mathcal{H}_\infty}J_{U_1}\\
=&2(1-\epsilon)\sum_{i=1}^{n_z+n_f}\frac{{\mathrm{Re}}(z_i)}{|z_i|^2}\sum_{j=1}^m\sigma_j^2\cos^2\angle(\eta_i,e_j)+J_{U_1}^*
,\end{align*}
where $C_{FM}$ is the minimum phase part of $FM$, $\eta_i$ is the direction vector associated
with the zero of $PF$ , $e_j$ is unitary a column vector, whose j-th element is 1 and the remaining elements 0, and
\begin{align*}
J_{U_1}^*=&\inf_{Q\in \mathbb{R}\mathcal{H}_\infty}\left\|\left\{
\begin{bmatrix}
\sqrt{1-\epsilon}I\\0
\end{bmatrix}
+\begin{bmatrix}
    -\sqrt{1-\epsilon}N_m \\
    \sqrt{\epsilon}C_{FM} \\
  \end{bmatrix}Q\right\}U\frac{1}{s}
\right\|_2^2\\
=&\inf_{Q\in \mathbb{R}\mathcal{H}_\infty}\left\|\left\{
\begin{bmatrix}
-\sqrt{1-\epsilon}I\\0
\end{bmatrix}
+\begin{bmatrix}
    \sqrt{1-\epsilon}N_m \\
    \sqrt{\epsilon}C_{FM} \\
  \end{bmatrix}Q\right\}U\frac{1}{s}
\right\|_2^2.
\end{align*}
Furthermore, we perform an inner-outer factorization given in \cite{Francis1987a}, such that
\begin{equation}\label{sqrtthetao}
\begin{bmatrix}
    \sqrt{1-\epsilon}N_m\\
    \sqrt{\epsilon}C_{FM}
\end{bmatrix}=\Theta_i\Theta_o,
\end{equation}
where $\Theta_i\in{\rm{\mathbb{R}\mathcal{H}_\infty}}$ is an inner matrix function, and $\Theta_o\in{\rm{\mathbb{R}\mathcal{H}_\infty}}$ is an outer. According to the definition of an inner matrix function, we have
\begin{equation}\label{thetaI}
\Theta_i^T(-j\omega)\Theta_i(j\omega)=I.
\end{equation}
From (\ref{sqrtthetao}), the following equation can be obtained
\begin{multline*}
\Theta_o^T(-j\omega)\Theta_o(j\omega)=(1-\epsilon)N_m^T(-j\omega)N_m(j\omega)\\
+\epsilon(C_{FM}^T(-j\omega)C_{FM}(j\omega)).
\end{multline*}
From (\ref{thetaI}), one can define the following matrix function with its module equal to 1,
\begin{align*}
\Psi(j\omega)=\begin{bmatrix}
\Theta_i^T(-j\omega)\\1-\Theta_i(j\omega)\Theta_i^T(-j\omega)
\end{bmatrix}.
\end{align*}
So according to the property of the matrix norm, ${\rm{J_{U_1}}}$ becomes
\begin{align*}
J_{U_1}^*=&\inf_{Q\in \mathbb{R}\mathcal{H}_\infty}\left\|\Psi\left\{
\begin{bmatrix}
-\sqrt{1-\epsilon}I\\0
\end{bmatrix}
+\begin{bmatrix}
    \sqrt{1-\epsilon}N_m \\
    \sqrt{\epsilon}C_{FM} \\
  \end{bmatrix}Q\right\}U\frac{1}{s}
\right\|_2^2\\
=&\inf_{Q\in \mathbb{R}\mathcal{H}_\infty}\left\|(A_1+\Theta_oQ)U\frac{1}{s}\right\|_2^2+\left\|A_2U\frac{1}{s}\right\|_2^2,
\end{align*}
where
\begin{align*}
A_1=&\Theta_i^H\begin{bmatrix}-\sqrt{1-\epsilon}I\\0\end{bmatrix}
=-(1-\epsilon)\,\Theta_o^{-H}N_m^H,\\
A_2=&(1-\Theta_i\Theta_i^H)\begin{bmatrix}-\sqrt{1-\varepsilon}I\\0\end{bmatrix}
=\begin{bmatrix}-\sqrt{1-\epsilon}I\\0\end{bmatrix}-\Theta_iA_1\\
=&\begin{bmatrix}
\sqrt{1-\varepsilon}\left(-I+(1-\epsilon)N_m\Theta_o^{-1}\Theta_o^{-H}N_m^H\right)\\
(1-\varepsilon)\sqrt{\epsilon}C_{FM}\Theta_o^{-1}\Theta_o^{-H}N_m^H
\end{bmatrix}.
\end{align*}
We then obtain
\begin{multline*}
{J}_{U_1}^*=(1-\epsilon)^2\bigg\|\big(\Theta_o^{-H}N_m^H\\
-\Theta_o^{-H}(0)N_m^H(0)\big)U\frac{1}{s}\bigg\|+\left\|A_2U\frac{1}{s}\right\|.
\end{multline*}
Similar to \cite{chen2003best}, we may invoke lemma \ref{le1},  and obtains
${J}_{U_1}^*=-(1-\epsilon)f^\prime(0).$
In light of lemma \ref{le2}, one also obtains
\[{J}_{U_1}^*=(1-\epsilon)(\sum_{i=1}^l{\sigma_i^2})\left[2\sum_{i=1}^{N_{s}}\frac{\mathrm{Re}s_i}{|s_i|}
-\frac{1}{\pi}\int_{-\infty}^{+\infty}\frac{\log{|f(j\omega)|}}{\omega^2}\mathrm{d}\omega\right].\]
Thus, we have
\begin{align*}
{J}_{U}^*
=&2(1-\epsilon)\Bigg[\sum_{i=1}^{n_z+n_f}\frac{{\mathrm{Re}}(z_i)}{|z_i|^2}\sum_{j=1}^l\sigma_j^2\cos^2\angle(\eta_i,e_j)\\
&+(\sum_{i=1}^l{\sigma_i^2})\left(\sum_{i=1}^{N_{s}}\frac{\mathrm{Re}s_{i}}{|s_{i}|^2}
-\frac{1}{\pi}\int_{0}^{+\infty}\frac{\log{|f(j\omega)|}}{\omega^2}\mathrm{d}\omega\right)\Bigg].
\end{align*}
Secondly, for $J_V$, we have
\begin{align*}
J_V=&\left\|
\begin{bmatrix}
    \sqrt{1-\epsilon}PM({\tilde{X}}-R{\tilde{N}})\\
    \sqrt{\epsilon}M({\tilde X}-R{\tilde N}) \\
  \end{bmatrix}HV
\right\|_2^2\\
=&\left\|
\begin{bmatrix}
    \sqrt{1-\epsilon}N_{om}\\
    \sqrt{\epsilon}M_m  \end{bmatrix}
    ({\tilde{X}}-R{\tilde{N}})HV
\right\|_2^2,
\end{align*}
where $N_{om}$ is the minimum phase part of $N_o$.

Similar to the equation (\ref{sqrtthetao}), we perform an inner-outer factorization such that
$$
\begin{bmatrix}
    \sqrt{1-\epsilon}N_{om}\\
    \sqrt{\epsilon}M_m
\end{bmatrix}=\Delta_i\Delta_0.
$$
In addition, similar to \cite{li2009optimal1}, we factorize $\tilde{N}HV=CD$, where $C$ is the minimum phase part, $D\in\mathbb{R}\mathcal{H}_\infty$ is an allpass factor which can be formed as
\begin{align}
\nonumber {D}(s):=&\prod\limits_{i=1}^{n_z+n_f}{{D_i}}(s),~~\\
\label{11}  {D_i}(s):= &{[\omega_i~~W_i]
\left[
  \begin{array}{cc}
    \frac{s-z_i}{s+\bar{z}_i} & 0 \\
    0 & I
  \end{array}
\right]
  \left[
    \begin{array}{c}
      \omega_i^H \\
      W_i^H
    \end{array}
  \right]}.
\end{align}
Hence, in light of Lemma \ref{le3}, we have
\begin{align*}
J_V=&\left\|
\Delta_0({\tilde{X}HV}-RCD)
\right\|_2^2
=\left\|
\Delta_0{\tilde{X}HV}D^{-1}-\Delta_0RC
\right\|_2^2\\
=&\bigg\|
\sum_{i=1}^{n_z+n_f}O(z_i)V{D}_{i}^l(z_i)[D_{i}^{-1}\\
&\qquad\qquad-D_{i}^{-1}(\infty)]{D}^r_i(z_i)+R_1-\Delta_0RC\bigg\|_2^2,
\end{align*}
where $R_1\in\mathbb{R}\mathcal{H}_\infty$, and
\begin{align*}
O(z_i)=&\Delta_0(z_i)\tilde{X}(z_i)H(z_i)
       =\Delta_0(z_i)M^{-1}(z_i)H(z_i),\\
{D}^l_i(z_i)=&D_1^{-1}(z_i)D_2^{-1}(z_i)\cdots{D_{i-1}^{-1}(z_i)},\\
{D}^r_i(z_i)=&D_{i+1}^{-1}(z_i)D_{i+2}^{-1}(z_i)\cdots{D_{n_z+n_f}^{-1}(z_i)}.
\end{align*}
Since $\Delta_0$ is right invertible and $C$ left invertible, we have
\begin{align*}
J_V^*=&\inf_{R\in \mathbb{R}\mathcal{H}_\infty}\Bigg\|\sum_{i=1}^{n_z+n_f}O(z_i)V{D}^l_{i}(z_i)[D_{i}^{-1}\\ &\qquad\qquad~-D_{i}^{-1}(\infty)]D^r_i(z_i)\Bigg\|_2^2+\left\|R_1-\Delta_0RC\right\|_2^2\\
   =&\left\|\sum_{i=1}^{n_z+n_f}O(z_i)V{D}^l_i(z_i)
\frac{2\mathrm{Re}(z_i)}{s-{z}_i}{\omega}_i
{\omega}_i^HD^r_i(z_i)\right\|_2^2\\
   =&\sum_{i,j=1}^{n_z+n_f}\frac{4Re(z_i)Re(z_j)}{\bar{z}_i+z_j}
   {\omega}_j^HD^r_i(z_j){D}_i^{rH}(z_i){\omega}_i\\
   &\qquad\qquad{\omega}_i^H{D_i^l}^H(z_i)V^HO^H(z_i)O(z_j)VD_j^l(z_j){\omega}_j.
\end{align*}
The proof is thus completed.
\end{proof}
\begin{remark}
When there is no network channel, because the Brownian motion random process is different from the step signal vector with deterministic direction, this result can not be degraded to the results in literature \cite{chen2003best}.
\end{remark}
\vspace{8pt}
\begin{corollary}
If the system $P(s)$ are SISO in Theorem \ref{th1}, then the optimal tracking performance can be written as
\begin{align*}
J^*=&2(1-\epsilon)\sigma^2\bigg[\sum_{i=1}^{n_z+n_f}\frac{{\mathrm{Re}}(z_i)}{|z_i|^2}
+\sum_{i=1}^{N_s}\frac{\mathrm{Re}s_{i}}{|s_{i}|^2}\\
&-\frac{1}{\pi}\int_{0}^{+\infty}\frac{\log{|f(j\omega)|}}{\omega^2}\mathrm{d}\omega\bigg]
+\gamma^2\sum_{i,j=1}^{n_z+n_f}\frac{\bar{\tilde{r}}_i\tilde{r}_j}{\bar{z}_i+z_j}.
\end{align*}
where
$$\tilde{r}_i={{\rm{Res}}_{s=z_i}}\Delta_o(z_i)M^{-1}(z_i)H(z_i)\hat{L}_i^{-1}$$
\end{corollary}
{\begin{figure}[h]
\centering
  \includegraphics[width=8.5cm]{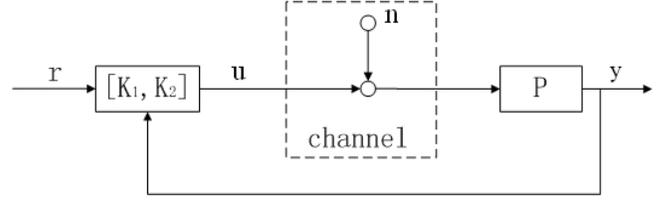}
  \caption{Feedback control over AWGN channels}\label{Fig.2}
\end{figure}}
\par
\begin{corollary}
Consider the simple channel case with $F=I$ and $H=I$. Under the assumptions in Theorem 1, define
$f(s):=\mathrm{tr}\{(1-\epsilon)U^TN_m(s)\Lambda_o^{-1}(s)\Lambda_o^{-T}(0)N_m^T(0)U\},$
and factorize $f(s):=\big(\prod_{i=1}^{N_s}{\bar{s}_i(s_i-s)}/{(s_i(\bar{s}_i+s))}\big)f_m(s),$
where $s_i\in\mathbb{C}_+$ are the nonminimum phase zeros of $f(s)$ and $f_m(s)$ is minimum phase. It is noted that, $f(s), f_m(s)\in\mathbb{R}\mathcal{H}\infty$, $ f(0)=f_m(0)=\sum_{i=1}^l{\sigma_i^2}.$
Then, with the two-parameter controller given in Fig.\ref{Fig.2},
\begin{align*}
J^*=&2(1-\epsilon)\Bigg[\sum_{i=1}^{n_z}\frac{{\mathrm{Re}}(z_i)}{|z_i|^2}
\sum_{j=1}^m\sigma_j^2\cos^2\angle(\eta_i,e_j)\\
&+(\sum_{i=1}^l{\sigma_i^2})\left(\sum_{i=1}^{N_{s}}\frac{\mathrm{Re}s_{i}}{|s_{i}|^2}-\frac{1}{\pi}\int_{0}^{+\infty}\frac{\log{|f(j\omega)|}}{\omega^2}\mathrm{d}\omega\right)\Bigg]\\   &+\sum_{i,j=1}^{n_z}\frac{4Re(z_i)Re(z_j)}{\bar{z}_i+z_j}\hat{\omega}_j^H\hat{D}^r_i(z_j)\hat{D}_i^{rH}(z_i)\hat{\omega}_i\\
&\times \hat{O}^H(z_i)\hat{O}(z_j)V\hat{D}_j^l(z_j)\hat{\omega}_j
\hat{\omega}_i^H{\hat{D}_i^{l~H}}(z_i)V^H.
\end{align*}
\end{corollary}
\begin{proof}
Similar to the proof of Theorem 1, we have the performance index
\begin{align*}
J: =&(1-\epsilon){\rm{tr}}\{{R_{\hat{e}_r}(0)}+{R_{y_n}(0)}\}
+\epsilon{\rm{tr}}\{{R_{u_{cr}}(0)}+{R_{u_{cn}}(0)}\}\\
   =&(1-\epsilon)(\|T_{\hat{e}_{r}}U\frac{1}{s}\|+\|T_{y_n}V\|)\\
   &\qquad\qquad\qquad\qquad\qquad\quad+\epsilon(\|T_{u_{cr}}U\frac{1}{s}\|+\|T_{u_{cn}}V\|)\\
   =&\left\|
\begin{bmatrix}
    \sqrt{1-\epsilon}(I-NQ) \\
    \sqrt{\epsilon}MQ \\
    \end{bmatrix}U\frac{1}{s}
\right\|_2^2\\
&\qquad\qquad\qquad\qquad\qquad+\left\|
\begin{bmatrix}
    \sqrt{1-\epsilon}N({\tilde{X}}-R{\tilde{N}})\\
    \sqrt{\epsilon}M({\tilde{X}}-R{\tilde{N}})\\
    \end{bmatrix}V
\right\|_2^2\\
   =&J_U+J_V.
\end{align*}
Let $\big[\sqrt{1-\epsilon}N_m^H~\sqrt{\epsilon}M_m^H\big]^H=\Lambda_i\Lambda_o,$
we can obtain
\begin{multline*}
{J}_{U}^*=2(1-\epsilon)\Bigg[\sum_{i=1}^{n_z}\frac{{\mathrm{Re}}(z_i)}{|z_i|^2}
\sum_{j=1}^n\sigma_j^2\cos^2\angle(\eta_i,e_j)\\
\qquad\quad+(\sum_{i=1}^l{\sigma_i^2})\left(\sum_{i=1}^{N_{s}}\frac{\mathrm{Re}s_{i}}{|s_{i}|^2}
-\frac{1}{\pi}\int_{0}^{+\infty}\frac{\log{|f(j\omega)|}}{\omega^2}\mathrm{d}\omega\right)\Bigg].
\end{multline*}
For $J_V$, we have
\begin{align*}
J_V=&\left\|
\begin{bmatrix}
    \sqrt{1-\epsilon}N({\tilde{X}}-R{\tilde{N}})\\
    \sqrt{\epsilon}M({\tilde{X}}-R{\tilde{N}})\\
\end{bmatrix}V\right\|_2^2\\
=&\left\|
\begin{bmatrix}
    \sqrt{1-\epsilon}N_m\\
    \sqrt{\epsilon}M_m
\end{bmatrix}({\tilde{X}}-R{\tilde{N}})V\right\|_2^2.
\end{align*}
In addition, we factorize $\tilde{N}V=\hat{C}\hat{D}$, where $\hat{C}$ is the minimum phase part, $\hat{D}\in\mathbb{R}\mathcal{H}_\infty$ is an allpass factor which can be formed as
\begin{align*}
\nonumber {\hat{D}}(s):=&\prod\limits_{i=1}^{n_z}{{\hat{D}_{i}}}(s),\\
  {\hat{D}_{i}}(s):=& {[\hat{\omega}_i~~\hat{W}_i]
\left[
  \begin{array}{cc}
    \frac{s-z_i}{s+\bar{z}_i} & 0 \\
    0 & I
  \end{array}
\right]
  \left[
    \begin{array}{c}
      \hat{\omega}_i^H \\ \hat{W}_i^H
    \end{array}
  \right]}.
\end{align*}
Similar to the proof of Theorem \ref{th1}, we can obtain
\begin{multline*}
J_V^*=\sum_{i,j=1}^{n_z}\frac{4Re(z_i)Re(z_j)}{\bar{z}_i+z_j}
   \hat{\omega}_j^H\hat{D}^r_i(z_j)\hat{D}_i^{rH}(z_i)\hat{\omega}_i\\
   \hat{\omega}_i^H{\hat{D}_i^{l~H}}(z_i)V^H
   \hat{O}^H(z_i)\hat{O}(z_j)V\hat{D}_j^l(z_j)\hat{\omega}_j,
\end{multline*}
where
\begin{align*}
\hat{O}(z_i)=&\Lambda_0(z_i)M^{-1}(z_i),\\
\hat{D}^l_i(z_i)=&\hat{D}_1^{-1}(z_i)\hat{D}_2^{-1}(z_i)\cdots{\hat{D}_{i-1}^{-1}(z_i)},~~\\
\hat{D}^r_i(z_i)=&\hat{D}_{i+1}^{-1}(z_i)\hat{D}_{i+2}^{-1}(z_i)\cdots{\hat{D}_{n_z}^{-1}(z_i)}.
\end{align*}
The proof is thus completed.\end{proof}
\par
If there is no channel noise in the configuration of the feedback control system depicted in Fig.2, then the following result can be immediately obtained.
\begin{corollary}\label{co2}
Consider the case of Fig.\ref{Fig.2}, and suppose that the channel is noise-free. Under the same assumptions described in Theorem 1, we have
\begin{align*}
{J}^*=&{J}_{U}^*=2(1-\epsilon)\bigg[\sum_{i=1}^{n_z}\frac{{\mathrm{Re}}(z_i)}{|z_i|^2}
\sum_{j=1}^l\sigma_j^2\cos^2\angle(\eta_i,e_j)\\
&+(\sum_{j=1}^l\sigma_j^2)\left(\sum_{i=1}^{N_{s}}\frac{\mathrm{Re}s_{fi}}{|s_{fi}|^2}
-\frac{1}{\pi}\int_{0}^{+\infty}\frac{\log{|f(j\omega)|}}{\omega^2}\mathrm{d}\omega\right)\bigg].
\end{align*}
\end{corollary}
\begin{remark}
If we do not consider the impact of the system control input, i.e. setting $\epsilon=0$, and $\sigma_j=1, (j=1,2,\cdots,l)$. From the expression in corollary \ref{co2}, it can be observed that for a feedback control system with a two-parameter compensators, when the tracking target is the Brownian  motion, the performance limitation depends on the nonminimum phase zeros, the plant gain at all frequencies
and their directions unitary vectors.
\end{remark}
\par
In what follows, we will discuss the relationship between stabilizability, the performance limits and channel characteristics under simplified conditions, we do some appropriate simplifications and assumptions. Consider the SISO system $P(s)$ in Fig.\ref{Fig.1} and the simplified performance $J$ as
\begin{equation}\label{eqq6}
J:=E\left\{\|r(t)-y_r(t)\|^2+\|y_n(t)\|^2\right\}.
\end{equation}
The relationship between the stabilizability, tracking performance and the channel signal-to-noise ratio (SNR) can be summarized as shown in the following theorem.
\begin{theorem}
Consider the feedback control system of Fig.{\ref{Fig.1}}. Suppose that $P(s)$ is a scalar transfer function. Under the assumptions in Theorem 1, the system $P(s)$ is stabilizable only if the admissible channel SNR satisfies
$$\frac{\mathcal{P}}{\gamma^2}>\sum_{i,j=1}^{n_p}\frac{{\bar{r}_i}{r_j}}{\bar{p}_i+p_j},$$
where $\mathcal{P}$ is the predetermined input power threshold. With the performance index (\ref{eqq6}), for the system to be stabilizable and obtain the optimal tracking performance, the channel SNR must satisfy
$$\frac{\mathcal{P}}{\gamma^2}>\sum_{i,j=1}^{n_p}\frac{{\bar{r}_i}{r_j}}{\bar{p}_i+p_j}+P_{Ad},$$
where
\begin{align}
{\rm{P_{Ad}}}=&\Bigg\|\sum_{i=1}^{n_p}N_{om}(p_i)N^{-1}(p_i)H(p_i)
\prod_{k=1,k\neq{i}}^{n_p}\tilde{B}^{-1}_k(p_i)\nonumber\\
&-\Bigg(\sum_{i=1}^{n_z+n_f}N_{om}(z_i)M^{-1}(z_i)H(z_i)\prod_{k=1,k\neq{i}}^{n_z+n_f}\hat{L}^{-1}_k(z_i)\nonumber\\ &+S\Bigg)N_m^{-1}\tilde{M}_m\Bigg\|_2^2,\label{eqp}
\end{align}
and the optimal tracking performance is given as
$$J^*=2\sigma^2\sum_{i=1}^{n_z+n_f}\frac{{\mathrm{Re}}(z_i)}{|z_i|^2}
+\gamma^2\sum_{i,j=1}^{n_z+n_f}\frac{\bar{\hat{r}}_i\hat{r}_j}{(\bar{z}_i+z_j)}
$$
where 
\begin{align*}
r_i=&{\rm{Res_{s=p_i}}}N_{om}(p_i)N^{-1}(p_i)H(p_i)\tilde{B}_i^{-1};\\
\hat{r}_i=&{\rm{Res_{s=z_i}}}N_{om}(z_i)M^{-1}(z_i)H(z_i)\hat{L}_i^{-1}.
\end{align*}
\end{theorem}
\begin{proof}
Using the equation (\ref{eqq6}), similar to the proof of the theorem 1, we have
\begin{align*}
 J:=&\rm{tr}\big(R_{\hat{e}_r}(0)+R_{y_n}(0)\big)
 =\bigg\|T_{\hat{e}_r}U\frac{1}{s}\bigg\|_2^2+\bigg\|T_{y_n}V\bigg\|_2^2\\
 =&\bigg\|(1-NQ)U\frac{1}{s}\bigg\|_2^2+\bigg\|PM(X-RN)HV\bigg\|_2^2\\
 =&\bigg\|\left[({L}^{-1}-1)+(1-N_mQ)\right]U\frac{1}{s}\bigg\|_2^2\\
 &\qquad\qquad\qquad\qquad\quad+\bigg\|N_{om}(X\hat{L}^{-1}-RN_m)HV\bigg\|_2^2\\
 =&\bigg\|({L}^{-1}-1)U\frac{1}{s}\bigg\|_2^2+\bigg\|(1-N_mQ)U\frac{1}{s}\bigg\|_2^2\\
 &\qquad\qquad\qquad\qquad\quad+\bigg\|N_{om}(X\hat{L}^{-1}-RN_m)HV\bigg\|_2^2
  \end{align*}\vspace{-8mm}
  \begin{align}
    \nonumber=&2\sigma^2\sum_{i=1}^{n_z+n_f}\frac{{\mathrm{Re}}(z_i)}{|z_i|^2}\\
   &\qquad\qquad\qquad+\bigg\|N_{om}(X\hat{L}^{-1}-RN_m)HV\bigg\|_2^2.\label{eq12}
\end{align}
Based on the allpass factorization (\ref{eq3}) and Lemma \ref{le3},
we can write
\begin{multline*}
N_{om}XH\hat{L}^{-1}=S+\sum_{i=1}^{n_z+n_f}N_{om}(z_i)\\
\times X(z_i)H(z_i)\hat{L}^{-1}(z_i)\hat{L}_i(z_i)\hat{L}_i^{-1},
\end{multline*}
where $S\in\mathbb{R}\mathcal{H}_\infty$.
Then
\begin{align}
\nonumber&\bigg\|N_{om}(X\hat{L}^{-1}-RN_m)HV\bigg\|_2^2\\
\nonumber=&\gamma^2\bigg\|R_1-N_{om}RN_mH+\sum_{i=1}^{n_z+n_f}N_{om}(z_i)X(z_i)\\
\nonumber&\qquad\qquad\times H(z_i)\hat{L}_i^{l}(z_i)\hat{L}_i^r(z_i)\left(\hat{L}_i^{-1}-\hat{L}_i^{-1}(\infty)\right)\bigg\|\\
\nonumber=&\gamma^2\bigg\|R_1-N_{om}RN_mH\bigg\|_2^2
+\gamma^2\bigg\|\sum_{i=1}^{n_z+n_f}O_1(z_i)\\
&\label{eq14}\qquad\quad\times \hat{L}_i^{l}(z_i)\hat{L}_i^r(z_i)\left(\hat{L}_i^{-1}-\hat{L}_i^{-1}(\infty)\right)\bigg\|
\end{align}
where
\par
$\qquad\quad\begin{array}{l}
         O_1(z_i)=N_{om}(z_i)X(z_i)H(z_i), \\
         \hat{L}^l_i(z_i)=\hat{L}_1^{-1}(z_i)\hat{L}_2^{-1}(z_i)\cdots{\hat{L}_{i-1}^{-1}(z_i)}, \\
         \hat{L}^r_i(z_i)=\hat{L}_{i+1}^{-1}(z_i)\hat{L}_{i+2}^{-1}(z_i)\cdots{\hat{L}_{n_z}^{-1}(z_i)},
       \end{array}$ \vspace{-2mm}
\begin{align}\label{eq16}
R_1(s)=&S+\sum_{i=1}^{n_z+n_f}O_1(z_i)\hat{L}_i^{l}(z_i)\hat{L}_i^r(z_i).
\end{align}
By using the Bezout identity
$XM-YN=1,$
$O_1(z_i)$ can be written as
\begin{equation}\label{eq15}
O_1(z_i)=N_{om}(z_i)M^{-1}(z_i)H(z_i).
\end{equation}
From equations (\ref{eq12}),(\ref{eq14}) and (\ref{eq15}), we have
\begin{align*}
J^*=&2\sigma^2\sum_{i=1}^{n_z+n_f}\frac{{\mathrm{Re}}(z_i)}{|z_i|^2}
+\gamma^2\sum_{i,j=1}^{n_z+n_f}\frac{4Re(z_i)Re(z_j)}{\bar{z}_i+z_j}\\
&\times \big(O_1(z_i)\hat{L}_i^l(z_i)\hat{L}_i^r(z_i)\big)^H O_1(z_i)\hat{L}_i^l(z_i)\hat{L}_i^r(z_i)\\
=&2\sigma^2\sum_{i=1}^{n_z+n_f}\frac{{\mathrm{Re}}(z_i)}{|z_i|^2}
+\gamma^2\sum_{i,j=1}^{n_z+n_f}O_1^H(z_i)O_1(z_i)\\
&\qquad\qquad\times \frac{4{\rm{Re}}(z_i){\rm{Re}}(z_j)}{(\bar{z}_i+z_j)} \prod_{k=1,k\neq{i}}^{n_z}\frac{(\bar{z}_k+{z}_i)(z_k+\bar{z}_j)}{(\bar{z}_k-\bar{z}_i)(z_k-z_j)}\\
=&2\sigma^2\sum_{i=1}^{n_z+n_f}\frac{{\mathrm{Re}}(z_i)}{|z_i|^2}
+\gamma^2\sum_{i,j=1}^{n_z+n_f}\frac{\bar{\hat{r}}_i\hat{r}_j}{(\bar{z}_i+z_j)}
\end{align*}
where $\hat{r}_i$ is the residue of $O_1(p_i)\hat{L}^{-1}(s)$ at $s=z_i$.\par
In addition, suppose that the input $r(t)=0$. The channel input is required to satisfy the power constraint
$\|u\|_{\rm{Pow}}<\mathcal{P}$
for some predetermined input power level $\mathcal{P}>0.$
\begin{align*}
\|u(t)\|_{\rm{Pow}}=&E[u^T(t)u(t)]={\rm{tr}}[R_{u_n}(0)]\\
\nonumber=&\int_{-\infty}^{+\infty}{\rm{tr}}\left(T_{u_n}S_n(jw)T_{u_n}^T\right){\rm{d}}w\\
\nonumber=&\|T_{u_n}V\|_2^2=\|N_{o}(\tilde{Y}-R\tilde{M})HV\|_2^2\\
\nonumber=&\gamma^2\|N_{om}\big(YH\tilde{B}^{-1}-RH\tilde{M}_m\big)\|_2^2\\
\nonumber=&\gamma^2\Bigg\|\sum_{i=1}^{n_p}N_{om}(p_i)Y(p_i)H(p_i)\prod_{k=1,k\neq{i}}^{n_p}\tilde{B}^{-1}_k(p_i)
\end{align*}
\begin{align}
\nonumber&\times [\tilde{B}_i^{-1}-\tilde{B}_i^{-1}(\infty)]\Bigg\|_2^2+
\gamma^2\Bigg\|\sum_{i=1}^{n_p}N_{om}(p_i)Y(p_i)H(p_i)\label{eq18}\\
&\times \prod_{k=1,k\neq{i}}^{n_p}\tilde{B}^{-1}_k(p_i)
-N_{om}RH\tilde{M}_m\Bigg\|_2^2.
\end{align}
When only the stabilizability is considered regardless of the tracking performance, we have
\begin{align}\label{eq19}
\nonumber&\|u(t)^*\|_{\rm{Pow\_S}}
=\gamma^2\Bigg\|\sum_{i=1}^{n_p}N_{om}(p_i)Y(p_i)H(p_i)\prod_{k=1,k\neq{i}}^{n_p}\tilde{B}^{-1}_k(p_i)\\
&\qquad\qquad\times [\tilde{B}_i^{-1}-\tilde{B}_i^{-1}(\infty)]\Bigg\|_2^2
=\gamma^2\sum_{i,j=1}^{n_p}\frac{{\bar{r}_i}{r_j}}{\bar{p}_i+p_j}
\end{align}
where $r_i={\rm{Res_{s=p_i}}}N_{om}(p_i)N^{-1}(p_i)H(p_i)\tilde{B}_i^{-1}$
is the residue of $N_{om}(p_i)N^{-1}(p_i)H(p_i)\tilde{B}_i^{-1}$ at $s=p_i$.
Therefore, for the feedback system to be stabilizable, the channel SNR must satisfy
\begin{align}\label{eq17}
\frac{\mathcal{P}}{\gamma^2}>\sum_{i,j=1}^{n_p}\frac{{\bar{r}_i}{r_j}}{\bar{p}_i+p_j}.
\end{align}
This is the result of \cite{rojas2008fundamental}. However, in many cases, not only the stabilizability needs to be considered, but also the system tracking performance. In this case, via noting equations (\ref{eq16}),(\ref{eq18}) and (\ref{eq19}), we have
\begin{align*}
&\|u^*(t)\|_{\rm{Pow\_SL}}=\gamma^2\Bigg\|\sum_{i=1}^{n_p}N_{om}(p_i)Y(p_i)H(p_i)\\
&\quad\times \prod_{k=1,k\neq{i}}^{n_p}\tilde{B}^{-1}_k(p_i)
-N_{om}RH\tilde{M}_m\Bigg\|_2^2+\gamma^2\sum_{i,j=1}^{n_p}\frac{{\bar{r}_i}{r_j}}{\bar{p}_i+p_j}\\
&\qquad\qquad\qquad=\|u^*(t)\|_{\rm{Pow\_S}}+\gamma^2{\rm{P_{Ad}}}.
\end{align*}
where ${\rm{P_{Ad}}}$ is given by equation (\ref{eqp}).
Then,
\begin{align*}
\frac{\mathcal{\mathcal{P}}}{\gamma^2}>\|u^*(t)\|_{\rm{Pow\_S}}/{\gamma^2}+{\rm{P_{Ad}}}.
\end{align*}
The proof is now completed.
\end{proof}
\begin{remark}
Theorem 2 shows that for a system to achieve for a best tracking performance in addition to stabilization, its signal-to-noise ratio must be greater than that required only for stabilization.
\end{remark}
\begin{figure}
\centering
  \includegraphics[width=8.8cm]{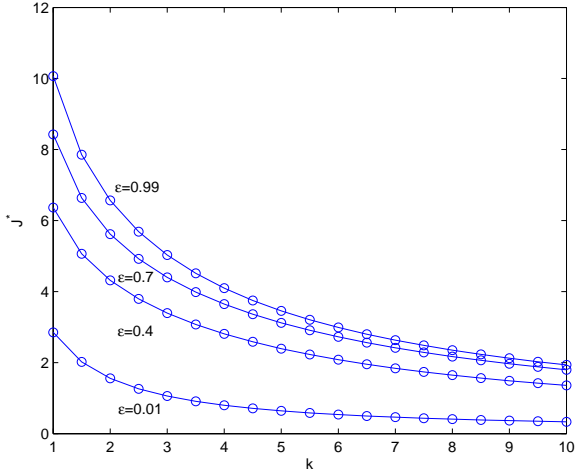}
  \caption{$J^*$ with respect to $k$ for different $\epsilon.(f=3,h=4,\sigma=1,\gamma=0.8)$}\label{Fig.3}
\end{figure}
\begin{figure}
\centering
  \includegraphics[width=8.8cm]{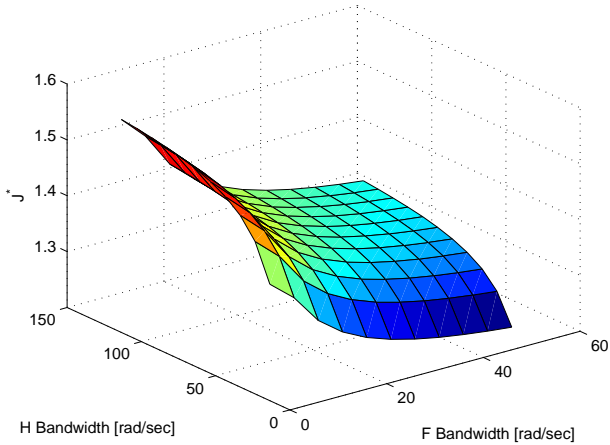}
  \caption{$J^*$ with respect to $F$ and $H.(k=2,\epsilon=0.5,\sigma=1,\gamma=0.8)$}\label{Fig.4}
\end{figure}
\begin{figure}
\centering
  \includegraphics[width=8.8cm]{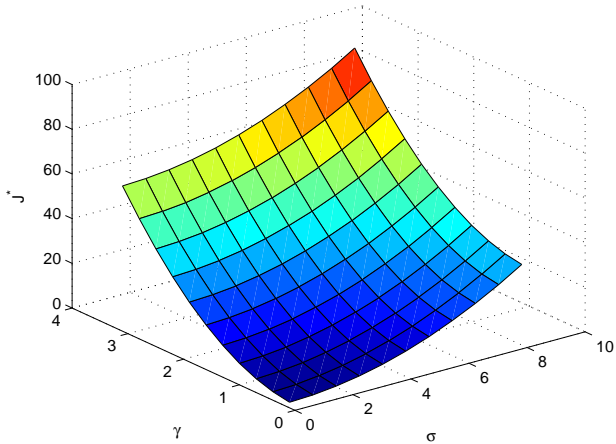}
  \caption{$J^*$ with respect to $\sigma$ and $\gamma.(k=2,\epsilon=0.5,f=3,h=4)$}\label{Fig.5}
\end{figure}

\section{Simulation studies}
Consider the plant
$$P={(s-k)}/{(s(s+1))}.$$
The LTI filters used to model the finite bandwidth $F(s)$ and colored noise $H(s)$ of the communication link are both chosen to be low-pass Butterworth filters of order 1.
$$F(s)=f/(s+f),~~~H(s)=h/(s+h),$$
where $k\in[1,10]$, and $f>0,~h>0.$

Clearly, $P(s)$ is of minimum phase. Fig.\ref{Fig.3} shows the optimal performances plotted for different values of $\epsilon$. Two observations can be obtained from Fig.\ref{Fig.4}, where the optimal performance is plotted with respect to bandwidth of both F(s) and H(s). First, the system tracking performance becomes better as the available bandwidth of the communication channel decreases. Secondly, if the noise is colored by a low pass filter, the decrease of its cutoff frequency would lead to the better tracking performance. Fig.\ref{Fig.5} shows that the reference signal and ACGN will deteriorate tracking performance.

\section{Conclusions}
In this paper, we have investigated the best attainable tracking performance of networked MIMO control systems in tracking the Brownian motion over a limited bandwidth and additive colored white Gaussian noise channel. We have derived explicit expressions of the best performance in terms of the tracking error and the control input energy. It has been shown that, due to the existence of the network, the best achievable tracking performance will be adversely affected by several factors, such as the nonminimum phase zeros and their directions of the plant, the colored additive white Gaussian noise, the basic network parameters, such as bandwidth. Finally, some simulation results are given to illustrate the obtained results.\par
Furthermore, one possible future work is to consider more realistic network-induced constraints, such as time-delay and dropout issues which is much more challenging. When the networked control system contains the nondeterministic or hybrid switching \cite{zhguan2005, Zhang2011Asynchronous, Zhang2010Necessary}, the issue of tracking performance also deserves further study.
\ifCLASSOPTIONcaptionsoff
  \newpage
\fi

%

\begin{IEEEbiography}[{\includegraphics[width=1in,height=1.25in,clip,keepaspectratio]{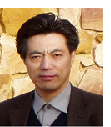}}]{Zhi-Hong Guan}
received PhD degree in Automatic Control Theory and Applications from the South China University of Technology, Guangzhou, China, in 1994. He was a Full Professor of Mathematics and Automatic Control with the Jianghan Petroleum Institute, Jingzhou, China in 1994. Since December 1997, he has been Full Professor of the Department of Control Science and Engineering, Executive Associate Director of the Centre for Nonlinear and Complex Systems and Director of the Control and Information Technology in the Huazhong University of Science and Technology (HUST), Wuhan, China.
\par
Since 1999, he has held visiting positions at the Harvard University, USA, the Central Queensland University, Australia, the Loughborough University, UK, the National University of Singapore, the University of Hong Kong, and the City University of Hong Kong. Currently, he is the Associate Editor of the "Journal of Control Theory and Applications", the international "Journal of Nonlinear Systems and Application", and severs as a Member of the Committee of Control Theory of the Chinese Association of Automation, Executive Committee Member and also Director of the Control Theory Committee of the Hubei Province Association of Automation. His research interests include complex systems and complex networks, impulsive and hybrid control systems, networked control systems, multi-agent systems.
\end{IEEEbiography}

\begin{IEEEbiography}[{\includegraphics[width=1in,height=1.25in,clip,keepaspectratio]{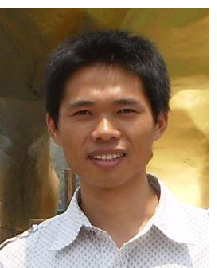}}]{Chao-Yang Chen} was born in Hunan, China, 1984. He graduated in mathematics from Hunan University of Science and Technology, Xiangtan, China, in 2003, and received the M.S. degree at Department of Mathematics in Guangxi Teachers Education University. Currently, he is working towards the Ph.D. Degree at the Department of Control Science and Engineering, Huazhong University of Science and Technology, Wuhan, China. His research interests include networked control systems, complex dynamical networks, impulsive and hybrid control systems.
\end{IEEEbiography}

\begin{IEEEbiography}[{\includegraphics[width=1in,height=1.25in,clip,keepaspectratio]{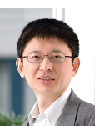}}]{Gang Feng} received the B.Eng and M.Eng. degrees in Automatic Control from Nanjing Aeronautical Institute, China in 1982 and in 1984 respectively, and the Ph.D. degree in Electrical Engineering from the University of Melbourne, Australia in 1992.
\par
He has been with City University of Hong Kong since 2000 where he is at present a Chair Professor. He is a ChangJiang Chair professor at Nanjing University of Science and Technology, awarded by Ministry of Education, China. He was lecturer/senior lecturer at School of Electrical Engineering, University of New South Wales, Australia, 1992-1999. He was awarded an Alexander von Humboldt Fellowship in 1997-1998, and the IEEE Transactions on Fuzzy Systems Outstanding Paper Award in 2007. His current research interests include piecewise linear systems, and intelligent systems \& control.
\par
Prof. Feng is an IEEE Fellow, an associate editor of IEEE Trans. on Fuzzy Systems, and was an associate editor of IEEE Trans. on Systems, Man \& Cybernetics, Part C, Journal of Control Theory and Applications, and the Conference Editorial Board of IEEE Control System Society.
\end{IEEEbiography}

\begin{IEEEbiography}[{\includegraphics[width=1in,height=1.25in,clip,keepaspectratio]{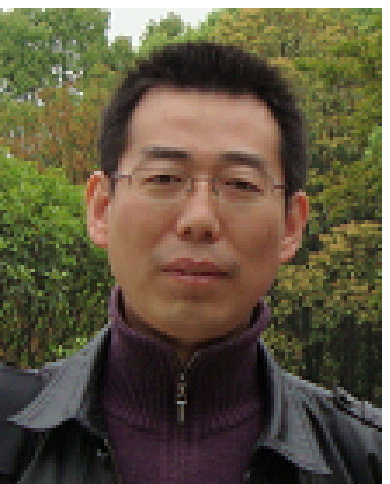}}]{Tao Li} received the Ph.D degree in Huazhong University of Science and Technology, Wuhan, China, in 2008. He is also currently an Associate Professor in the College of Electronics and Information, Yangtze University, Jingzhou, China. His current research interests include nonlinearity complex network systems, complex network theory \& application, complex networks spreading dynamics.
\end{IEEEbiography}






\begin{thebibliography}{1}

\bibitem{braslavsky2007feedback}
J.H.~Braslavsky, R.H.~Middleton, and J.S.~Freudenberg, ``Feedback stabilization over signal-to-noise ratio constrained channels," \emph{IEEE Transactions on Automatic Control}, vol. 52, no. 8, pp. 1391--1403, 2007.

\bibitem{li2009optimal}
Y.~Li, E.~Tuncel, and J.~Chen, ``Optimal tracking over an additive white gaussian noise channel," \emph{Proceedings of the 2009 American Control Conference}, MO, USA, 2009, pp. 4026--4031.

\bibitem{xisheng2010performance}
X.~Zhan, Z.~Guan, J.~Xiao, and Y.~Wang, ``Performance limitations in tracking of linear system with measurement noise," \emph{Proceedings of the 29nd IEEE Conference on Chinese Control Conference}, Beijing, China, 2010, pp. 1614--1617.

\bibitem{rojas2008fundamental}
A.J.~Rojas, J.H.~Braslavsky, and R.H.~Middleton, ``Fundamental limitations in control over  a communication channel," \emph{Automatica}, vol. 44, no. 12, pp. 3147--3151, 2008.

\bibitem{xiao2010feedback}
N.~Xiao and L.~Xie, ``Feedback stabilization over stochastic multiplicative input channels: Continuous-time case,"  \emph{Proceedings of the 11th International Conference Control, Automation, Robotics and Vision}, Singapore, 2010, pp. 543--548.

\bibitem{menon2010static}
P.P.~Menon and C.~Edwards, ``Static output feedback stabilisation and synchronisation of complex networks with H2 performance," \emph{International Journal of Robust and Nonlinear Control}, vol. 20, no. 6, pp. 703--718, 2010.

\bibitem{guan2011optimal}
Z.H.~Guan, X.S.~Zhan, and G.~Feng, ``Optimal tracking performance of mimo discrete-time systems with communication constraints," \emph{International Journal of Robust and Nonlinear Control}, Wiley Online Library.

\bibitem{Zhang2011Quantized}
H. Zhang, H. Yan, F. Yang, and Q. Chen, ``Quantized Control Design for Impulsive Fuzzy Networked Systems," \emph{IEEE Transactions on Fuzzy Systems}, vol. 19, no. 6, pp. 1153-1162, 2011.

\bibitem{Azuma2012Dynamic}
SI Azuma and T. Sugie, ``Dynamic Quantization of Nonlinear Control Systems," \emph{IEEE Transactions on Automatic Control}, vol. 57, no. 4, pp. 875-888, 2012.

\bibitem{qi2009optimal}
T.~Qi and W.~Su, ``Optimal tracking and tracking performance constraints from quantization," \emph{Proceedings of the 7th Asian Control Conference}, China, 2009, pp. 447--452.

\bibitem{you2009optimality}
K.~You, W.~Su, M.~Fu, and L.~Xie, ``Optimality of the logarithmic quantizer for stabilization of linear systems: Achieving the minimum data rate," \emph{Proceedings of the 48th IEEE Conference on Decision and Control and 28th Chinese Control Conference}, China, 2009, pp. 4075--4080.

\bibitem{Xiao2010stabilization}
N. Xiao, L. Xie, and M. Fu, ``Stabilization of Markov jump linear systems using quantized state feedback," \emph{Automatica}, vol. 46, no. 10, pp. 1696--1702, 2010.

\bibitem{Luan2011Stabilization}
X. Luan, P. Shi, and F. Liu, ``Stabilization of Networked Control Systems with Random
Delays," \emph{IEEE Transactions on Industrial Electronics}, vol. 58, no. 9, pp. 4323-4330, 2011.

\bibitem{Wei2009Filtering}
G. Wei, Z. Wang, X. He, and H. Shu, ``Filtering for Networked Stochastic Time-Delay Systems With Sector Nonlinearity," \emph{IEEE Transactions on Circuits and Systems II:Express Briefs}, vol. 56, no. 1, pp. 71--75, 2009.

\bibitem{Liu2010Predictive}
G.P. Liu, ``Predictive Controller Design of Networked Systems With Communication Delays and Data Loss," \emph{IEEE Transactions on Circuits and Systems II: Express Briefs}, vol. 57, no. 6, pp. 481--485, 2010.

\bibitem{rojas2006output}
A.J.~Rojas, J.H.~Braslavsky, and R.H.~Middleton, ``Output feedback stabilisation over bandwidth limited, signal to noise ratio constrained communication channels," \emph{Proceedings of the 29th American Control Conference}, Minnesota, USA, 2006, pp. 14--16.

\bibitem{Trivellato2010State}
M. Trivellato and N. Benvenuto, ``State Control in Networked Control Systems under Packet Drops and Limited Transmission Bandwidth," \emph{IEEE Transactions on Communications}, vol. 58, no. 2, pp. 611--622, 2010.

\bibitem{wu2007design}
J.~Wu and T.~Chen, ``Design of networked control systems with packet dropouts," \emph{IEEE Transactions on Automatic  Control}, vol. 52, no. 7, pp. 1314--1319, 2007.

\bibitem{Wang2011a}
Y. Wang, W. Liu, X. Zhu, and Z. Du, ``A survey of networked control systems with delay and packet dropout," \emph{Proceedings of 2011 Chinese Control and Decision Conference (CCDC)}, 2011, pp. 2342--2346.

\bibitem{you2011mean}
K. You, M. Fu, and L. Xie, ``Mean square stability for Kalman filtering with Markovian packet losses," \emph{Automatics}, vol.47, no.12, pp.2647--2657, 2011.

\bibitem{toker2002tracking}
O.~Toker, J.~Chen, and L. Qiu, ``Tracking performance limitations in LTI multivariable discrete-time systems," \emph{IEEE Transactions on Circuits and Systems I, Fundamental Theory and Applications}, vol. 49, no. 5, pp. 657--670, 2002.

\bibitem{chen2003best}
J.~Chen, S.~Hara, and G.~Chen, ``Best tracking and regulation performance under control energy constraint," \emph{IEEE Transactions on Automatic Control}, vol. 48, no. 8, pp. 1320--1336, 2003.

\bibitem{bakhtiar2008regulation}
T.~Bakhtiar and S.~Hara, ``Regulation performance limitations for simo linear time-invariant feedback control systems," \emph{Automatica}, vol. 44, no. 3, pp. 659--670, 2008.

\bibitem{wang2011optimal}
B.X.~Wang, Z.H.~Guan, and F.S.~Yuan, ``Optimal tracking and two-channel disturbance rejection under control energy constraint," \emph{Automatica}, vol. 47, no. 4, pp. 733--738, 2011.

\bibitem{morari1989robust}
M.~Morari and E.~Zafiriou, \emph{Robust process control}. Englewood Cliffs. NJ:Prentice-Hall, 1989.

\bibitem{chen2000limitations}
J.~Chen, L.~Qiu, and O.~Toker, ``Limitations on maximal tracking accuracy," \emph{IEEE Transactions on Automatic  Control}, vol. 45, no. 2, pp. 326--331, 2000.

\bibitem{ding2010tracking}
L.~Ding, H.~Wang, Z.~Guan, and J.~Chen, ``Tracking under additive white gaussian noise effect," \emph{IET Control Theory and Application}, vol. 4, no. 11, pp. 2471--2478, 2010.

\bibitem{qiu2002fundamental}
L.~Qiu, Z.~Ren, and J.~Chen, ``Fundamental performance limitations in estimation problems," \emph{Communications in Information and Systems}, vol. 2, no. 4, pp. 371--384, 2002.

\bibitem{Zhan2012Optimal}
X.S. Zhan , Z.H. Guan , R.Q. Liao , F.S Yuan. ``Optimal Performance in Tracking Stochastic Signal under Disturbance Rejection". \emph{Asian Journal of Control}. DOI: 10.1002/asjc.494.

\bibitem{wang2009limitations}
H.~Wang, L.~Ding, Z.~Guan, and J.~Chen, ``Limitations on minimum tracking energy for siso plants," \emph{Proceedings of the Control and Decision Conference}, China, 2009, pp. 1432--1437.

\bibitem{li2009optimal1}
Y.~Li, E.~Tuncel, and J.~Chen, ``Optimal tracking and power allocation over an additive white noise channel," \emph{Proceedings of the IEEE International Conference on Control and Automation}, New Zealand, 2009, pp. 1541--1546.

\bibitem{Francis1987a}
B.A. Francis, \emph{A Course in H$_\infty$ Control Theory}, ser. Lecture Notes in Control and Information Science. Berlin, Germany: Springer-Verlag,
1987.

\bibitem{zhguan2005}
Z. H. Guan, D. J. Hill, and X. Shen.  On hybrid impulsive and switching systems and application to nonlinear control.
 \emph{IEEE Transactions on Automatic and Control}, vol. 50, no. 7, pp.  1058-1062, 2005.

\bibitem{Zhang2011Asynchronous}
L. Zhang, N. Cui, M. Liu, and Y. Zhao, ``Asynchronous Filtering of Discrete-Time Switched Linear Systems With Average Dwell Time," \emph{IEEE Transactions on Circuits and Systems I: Regular Papers}, vol. 58, no. 5, pp. 1109--1118, 2011.

\bibitem{Zhang2010Necessary}
L. Zhang and L. James , ``Necessary and Sufficient Conditions for Analysis and Synthesis of Markov Jump Linear Systems With Incomplete Transition Descriptions," \emph{IEEE Transactions on Automatic and Control}, vol. 55, no. 7, pp. 1695--1701, 2010.

\end{thebibliography}
\end{document}